\documentclass[12pt]{article}
\usepackage{amsmath,amssymb,amsthm}
\usepackage{graphicx}
\newtheorem{theorem}{Theorem}

\newtheorem{conclusion}[theorem]{Conclusion}

\begin{document}

\title{Microscopic Unitarity and the Quantization of Black Hole Evaporation Time}
\author{Ahmad Adel Abutaleb \\
Mansoura University, Faculty of Science, Egypt}
\date{}
\maketitle

\begin{abstract}
This work presents an effective microscopic, time-dependent Hamiltonian framework for investigating information dynamics during black hole evaporation. While current approaches often rely on gravitational path integrals or statistical ensembles to recover the Page curve, our model provides an explicit, unitary quantum-mechanical evolution of the radiation. We utilize the Independent Unitary Pairing assumption, where the total Hilbert space is decomposed into $N_{\text{total}}$ mutually independent bipartite subsystems. Each subsystem consists of a single interior qubit interacting unitarily with a single radiation qubit, ensuring strict global microscopic unitarity (Von Neumann entropy $=0$) throughout the process. To reconcile this with macroscopic thermodynamics, we introduce a methodology called ``Fermion-like Occupancy Bound'', which can be considered as a Holographic Binary Capacity Constraint, where each radiation channel is modeled as a two-level system representing the fundamental unit of information. This truncation, justified by the holographic principle at the Planck scale, enforces a maximum entropy bound of $\ln 2$ per channel, which naturally yields the entropy turnaround and final state purification.

The central result of this framework is the derivation of a Quantum Condition for Unitarity ($\theta _{j,a}(\mathbf{t_{\text{evap}}})=n_{j,a}\cdot \frac{\pi }{2}$), which couples microscopic phase evolution with macroscopic observables. By combining this condition with the semiclassical Hawking mass-loss law, we establish a fundamental scaling relation between the initial black hole mass $M_{0}$, the coupling strength $\alpha$ and the total evaporation time $\mathbf{t_{\text{evap}}}$. The results suggest a potential quantization of the evaporation time, where this evaporation time is expressed as discrete multiples of a fundamental temporal unit. Numerical simulations corroborate the theoretical framework, demonstrating the emergence of the Page-like curve and providing a mechanism for how information is preserved.
\end{abstract}

\section{Introduction}
A black hole is defined as a region of spacetime with extremely high gravity, resulting from the collapse of a massive object. In the 1970s, some physicists (including Hawking) formulated four laws of black hole mechanics. The second law, also known as the event horizon area law, states that the area of the event horizon of any black hole must either increase or remain constant over time. Bekenstein noted the similarity between this law and the second law of thermodynamics, which states that the entropy of any closed system can never decrease; it always increases or at least remains constant. This led him to hypothesize that black holes possess entropy proportional to the area of their event horizon, and that most of the universe's entropy is held by black holes. When Hawking considered the effects of quantum mechanics near the event horizon ``Quantum Vacuum Fluctations'', he found that, as a result of the black hole's immense gravity, one of the two particles from the hypothetical particle-antiparticle pairs carrying negative energy is pulled in, leading to a decrease in the black hole's mass and its ``evaporation,'' while the other particle manages to escape, transforming into a real particle that can be detected, known as Hawking radiation. This led to what is called the black hole information paradox \cite{Bekenstein1973,Hawking1975,Wald2001}.

The black hole information paradox arises from the apparent tension between the unitary evolution required by quantum mechanics and the thermal character of Hawking radiation predicted by semiclassical calculations \cite{Hawking1976}: if an initially pure state collapses to form a black hole and the Hawking radiation remains exactly thermal, then the final radiation state would be mixed and the information would be lost. Don Page addressed the issue of information loss by studying the average entropy of subsystems of a random pure state. He predicted the characteristic `Page curve' for the radiation entropy \cite{Page1993Info,Page1993Avg,Almheiri2017Overview}, which initially rises, peaks near the halfway (Page) time, and then returns to zero as evaporation completes.

Over the past two decades quantum-information-theoretic toy models have clarified aspects of how information might be returned to the radiation. In particular, Hayden and Preskill showed that a rapidly scrambling, unitary black hole behaves like an efficient quantum encoder, i.e., information dumped into an old black hole can be retrieved quickly from subsequent Hawking radiation, illustrating how unitarity can be compatible with rapid information release under appropriate assumptions about internal dynamics \cite{Hayden2007}. Simultaneously, analyses of entanglement monogamy, the principle that a quantum particle cannot be fully entangled with two different particles at the same time, led to the so called ``firewall'' paradox \cite{Almheiri2013}, which sharpened the conceptual tension by arguing that preserving both unitarity and the smoothness of the horizon necessitates mutually incompatible entanglement relations between late radiation, early radiation, and interior partners.

Recent progress has taken place on two complementary fronts. From the gravitational path-integral side, developments involving replica wormholes and the quantum extremal surface (island) prescription have produced semiclassical calculations that reproduce Page-like behavior for evaporating black holes \cite{Penington2020,Almheiri2019}, thereby showing that semiclassical gravity (when carefully treated with certain replica saddles) can in some contexts yield unitary-like entropy curves for the radiation. These results provide an appealing semiclassical mechanism for the Page transition and have clarified the role of new topologies in the gravitational path integral. On the other hand, many microscopic or information-theoretic demonstrations of Page-like behavior rely on statistical ensembles (e.g., Haar-random unitaries or random states) or on idealized decoupling arguments rather than on explicit, time dependent Hamiltonians acting on well defined physical states. However, a gap remains in establishing an explicit, time-dependent, and microscopically unitary quantum-mechanical model that reproduces the Page-like curve while connecting the quantum dynamics directly to the macroscopic black hole evaporation law. Previous microscopic models often rely on averaged states or large ensembles rather than an explicit Hamiltonian formulation.

Our work addresses this by presenting a self-consistent, Hamiltonian-based unitary framework. This framework explicitly demonstrates how microscopic unitarity can be perfectly maintained while statistically generating the macroscopic thermal properties and the required information turnaround, thus providing a microscopic demonstration of the information preservation during black hole evaporation.

The foundation of this framework is the \textit{Independent Unitary Pairing} assumption, that is the total Hilbert space decomposes into $N_{\text{total}}$ mutually independent bipartite two-qubit subsystems. This microscopic decomposition guarantees an exact factorization of the total radiation entropy:
\begin{equation*}
S_{\text{Rad}}^{\text{Total}} = \sum_{j,a} S_{j,a}
\end{equation*}
where each contribution $S_{j,a}(t)$ is computed rigorously as the von Neumann entanglement entropy of the corresponding reduced density matrix. Because each elementary interaction is unitary, the global von Neumann entropy of the combined system remains strictly zero throughout the evolution:
\begin{equation*}
S_{\text{vN}}(\rho_{\text{total}}) = 0
\end{equation*}
ensuring that information is never lost at the microscopic level.

A key conceptual element of our construction is inspired by ideas appearing in the literature \cite{Osuga2016,Broda2020,Chu2024}. We consider each radiation mode as an effective information carrier instead of focusing on the actual number of particles. To satisfy the necessary mathematical condition for unitarity and to reproduce the Page curve, the Hilbert space of each radiation mode is truncated to a \textit{Two-level Qubit System}.

This representation defines a binary information state: the state $|0\rangle$ represents the unexcited mode, carrying no information/entanglement originating from the black hole, while the state $|1\rangle$ represents the excited/occupied mode, carrying a single unit of quantum information. Consequently, the physical details related to the true multi-boson occupation ($n_{j}=1,2,3,\ldots $) become irrelevant, as any excitation is represented by $|1\rangle$. We refer to this methodology as the ``Fermion-like Occupancy Bound'', motivated by the holographic principle, where each fundamental Planck-area cell on the horizon acts as a discrete bit of information, which imposes a maximum entropy of $\ln 2$ for each microscopic channel. This maximum limit is what ensures that the total entropy of the radiation is bounded and can reach a turning point and begin to decrease (purification), which is the essential condition for achieving unitary evolution in quantum mechanics. Thus, the assumption of the binary representation $\{|0\rangle ,|1\rangle \}$ for each mode does not contradict the physical fact that Hawking radiation is composed of bosons, which can occupy the same mode in unlimited numbers.

We suggest a bridge between microscopic quantum dynamics and macroscopic thermodynamic observables. For each microscopic radiation channel, the quantum excitation probability is shown to be mathematically identical to the expectation value of its occupation number. This identity allows the unbounded occupation characteristic of bosonic Hawking radiation to emerge naturally from the statistical aggregation of many independent, yet strictly binary-bounded, qubit channels---an identification justified by an ergodic coarse-graining argument. The macroscopic radiation mode at frequency $\omega _{j}$ is therefore modeled as the collective ensemble of its constituent microscopic channels.

Despite the rigidity of microscopic unitarity, macroscopic entropy growth arises statistically as $S_{\text{Rad}}^{\text{Total}}(t)$ accumulates through coarse-graining of continuously increasing microscopic entanglement. To guarantee a physically valid entropy turnaround, the model incorporates a holographic interpretation in which the black hole horizon area is effectively quantized into $N_{\text{total}}$ fundamental information-carrying channels. Each channel behaves as a qubit and thus enforces a binary occupancy bound. This bound sets an upper limit on the entropy of each macroscopic radiation mode:
\begin{equation*}
S_{j}^{\text{max}}=N_{j}\ln 2,
\end{equation*}
where $N_j$ denotes the number of microscopic channels associated with frequency $\omega_j$. This maximum is essential: it ensures that the radiation entropy necessarily reaches a peak and subsequently decreases---reproducing the characteristic Page curve expected from unitary evaporation.

The most important outcome of this framework is the derivation of a \textit{Quantum Condition for Unitarity}. Requiring the final total radiation entropy to return strictly to zero imposes a quantization rule on the integrated unitary phase of each microscopic channel:
\begin{equation*}
\theta _{j,a}(\mathbf{t_{\text{evap}}})=n_{j,a}\cdot \frac{\pi }{2},\quad n_{j,a}\in \mathbb{Z}.
\end{equation*}
By combining this quantization condition with the classical Hawking mass-loss law, the model yields an explicit relation connecting the initial black hole mass $M_0$, the microscopic coupling strengths $\alpha_{j,a}$, and the total evaporation time $t_{\text{evap}}$. This result implies not only a potential quantization of the evaporation time but also a strict scaling relation:
\begin{equation*}
\alpha_{j,a} \propto n_{j,a}
\end{equation*}
between the microscopic coupling and the final accumulated quantum phase---providing a nontrivial constraint on the information-processing dynamics at the event horizon.

In summary, the framework demonstrates how deeply quantum, microscopically unitary dynamics can reproduce some macroscopic thermodynamic features of black hole evaporation.

\section{The Microscopic Unitary Framework}
In this framework, we set $\hbar =c=k_{B}=G=1$. The microscopic unitary model is built upon the assumption of \textbf{Independent Unitary Pairing} to guarantee the factorization of the total radiation entropy. We postulate that the total system comprises $N_{\text{total}}$ microscopic pairs, where each pair consists of one interior degree of freedom and one radiation channel.

Each microscopic pair is indexed by $(j,a)$, where: $j=1,\dots ,J$ indexes the frequency mode ($\omega _{j}$) and $a=1,\dots ,N_{j}$ indexes individual channels within mode $j.$ Crucially, each microscopic radiation channel $\text{rad}_{j,a}$ is coupled exclusively to its own unique interior qubit $\text{in}_{j,a}$, and the interactions are non-overlapping in the interior Hilbert space. This assumption simplifies the total system into $N_{\text{total}}$ mutually independent bipartite systems.

We model the interaction for a single microscopic pair $(j,a)$ as a two-qubit system evolving under a time-dependent Hamiltonian $H_{j,a}(t)$. The Hilbert space of this pair is:
\begin{equation}
\mathcal{H}_{j,a}^{\text{total}}=\mathcal{H}_{\text{in}_{j,a}}\otimes \mathcal{H}_{\text{rad}_{j,a}}.
\end{equation}
At time $t=0$, the joint system is prepared in the vacuum state:
\begin{equation}
|\Psi _{j,a}(0)\rangle =|0\rangle _{\text{in}}\otimes |0\rangle _{\text{rad}}\equiv |00\rangle .
\end{equation}
The time evolution is governed by an interaction Hamiltonian of the form:
\begin{equation}
H_{j,a}(t)=g_{j,a}(t)A,\quad \text{where }A=\sigma _{x}\otimes \sigma _{x},
\end{equation}
and the coupling function is defined as:
\begin{equation}
g_{j,a}(t)=\frac{\alpha _{j,a}}{M(t)^{2}},
\end{equation}
where $M(t)$ is the black hole mass and $\alpha _{j,a}$ is the coupling strength.

The unitary evolution operator is defined using the time-ordering operator $\mathcal{T}$,
\begin{equation}
U_{_{j,a}}(t)=\mathcal{T}\exp \left( -i\int_{0}^{t}H_{_{j,a}}(\tau )d\tau \right).
\end{equation}
Since the Hamiltonian commutes with itself at all times, i.e.,
\begin{equation}
\lbrack H_{j,a}(t_{1}),H_{j,a}(t_{2})]=g_{j,a}(t_{1})g_{j,a}(t_{2})[A,A]=0,
\end{equation}
the unitary evolution operator simplifies to:
\begin{equation}
U_{j,a}(t)=\exp \left( -i\int_{0}^{t}g_{j,a}(\tau )d\tau \cdot A\right).
\end{equation}
Defining the integrated phase as:
\begin{equation}
\theta _{j,a}(t)=\int_{0}^{t}g_{j,a}(\tau )d\tau,
\end{equation}
we have
\begin{equation}
U_{_{j,a}}(t)=e^{-i\theta _{_{j,a}}(t)A}.
\end{equation}
The Taylor series expansion for the exponential can be written as:
\begin{equation}
e^{-i\theta _{_{j,a}}A}=\sum_{k=0}^{\infty }\frac{(-i\theta _{_{j,a}}A)^{k}}{k!}=\sum_{m=0}^{\infty }\frac{(-i\theta _{_{j,a}}A)^{2m}}{(2m)!}+\sum_{m=0}^{\infty }\frac{(-i\theta _{_{j,a}}A)^{2m+1}}{(2m+1)!}.
\end{equation}
Using $A^{2m}=I$ and $A^{2m+1}=A$, we get:
\begin{equation}
e^{-i\theta _{_{j,a}}A}=I\sum_{m=0}^{\infty }\frac{(-1)^{m}\theta _{j,a}^{2m}}{(2m)!}-iA\sum_{m=0}^{\infty }\frac{(-1)^{m}\theta _{j.a}^{2m+1}}{(2m+1)!}.
\end{equation}
Recognizing the cosine and sine series, we obtain:
\begin{equation}
U_{j,a}(t)=\cos (\theta _{j,a}(t))I-i\sin (\theta _{j,a}(t))A.
\end{equation}
Now, when we apply the unitary operator to the initial state, we get:
\begin{equation}
|\Psi _{j,a}(t)\rangle =U_{j,a}(t)|00\rangle =\cos (\theta _{j,a}(t))|00\rangle -i\sin (\theta _{j,a}(t))|11\rangle.
\end{equation}
The reduced density matrix for the radiation subsystem is found by taking the partial trace over the interior degree of freedom, i.e.,
\begin{equation}
\rho _{\text{rad},j,a}(t)=\operatorname{Tr}_{\text{in}}[|\Psi _{j,a}(t)\rangle\langle \Psi _{j,a}(t)|].
\end{equation}
The partial trace yields the diagonal matrix:
\begin{equation}
\rho _{\text{rad},j,a}(t)=\cos ^{2}(\theta _{j,a}(t))|0\rangle \langle 0|+\sin ^{2}(\theta _{j,a}(t))|1\rangle \langle 1|.
\end{equation}
It is clear that the eigenvalues are $\lambda _{1}=\cos ^{2}(\theta _{j,a}(t))$ and $\lambda _{2}=\sin ^{2}(\theta _{j,a}(t))$.

We define the excitation probability as the eigenvalue corresponding to the excited state $|1\rangle$, i.e.
\begin{equation}
p_{j,a}(t)=\sin ^{2}\left( \theta _{j,a}(t)\right) =\sin ^{2}\left( \int_{0}^{t}g_{j,a}(\tau )d\tau \right).
\end{equation}
Now, the entanglement entropy for the microscopic channel is derived using the von Neumann entropy formula:
\begin{equation}
S_{j,a}(t)=-\operatorname{Tr}[\rho _{\text{rad},j,a}(t)\ln \rho _{\text{rad},j,a}(t)],
\end{equation}
so, for this two-level system, this simplifies to the binary entropy formula:
\begin{equation}
S_{j,a}(t)=-p_{j,a}(t)\ln p_{j,a}(t)-(1-p_{j,a}(t))\ln (1-p_{j,a}(t)).
\end{equation}

\section{From Microscopic to Macroscopic}
We define a ``macroscopic channel $j$'' as a collection of microscopic channels, where each microscopic channel has a frequency $\omega _{j}$, and the number of these microscopic channels is $N_{j}.$ For this macroscopic channel, we define:
\begin{equation}
p_{j}(t)=\frac{1}{N_{j}}\sum_{a=1}^{N_{j}}p_{j,a}(t),
\end{equation}
where $p_{j,a}(t)=\sin ^{2}(\theta _{j,a}(t))$ is the microscopic quantum excitation probability. So, $p_{j}(t)$ can be considered as the average occupation ratio for the macroscopic channel $j$ (which contains $N_{j}$ microscopic channels).

Also, we define the total occupied number of quanta $n_{j}^{(\text{macro})}(t)$ as follows:
\begin{equation}
n_{j}^{(\text{macro})}(t)=\sum_{a=1}^{N_{j}}n_{j,a}^{(\text{micro})}(t),
\end{equation}
where $n_{j,a}^{(\text{micro})}(t)$ represents the actual state of qubit number $a$ at time $t.$ Since it is a qubit, it is a binary quantum system with two fundamental states: the ground state $|0\rangle $ and the excited state $|1\rangle$, so $n_{j,a}^{(\text{micro})}(t)\in \{0,1\}$, where $n_{j,a}^{(\text{micro})}(t)=1$ means that qubit $a$ which has the frequency $\omega _{j}$ is in the excited state $|1\rangle$, and $n_{j,a}^{(\text{micro})}(t)=0$ means that qubit $a$ is in the ground state $|0\rangle$.

\begin{theorem}
The expectation value of $n_{j,a}^{(\text{micro})}(t)$,
\begin{equation}
\mathbb{E}[n_{j,a}^{(\text{micro})}(t)]=p_{j,a}(t).
\end{equation}
\end{theorem}

\begin{proof}
Since in the microscopic formulation, each radiative mode $j,a$ is modeled as a two-level quantum system (a qubit), with basis states $|0\rangle $ (ground state) and $|1\rangle $ (excited state), we can define the number operator associated with this mode as:
\begin{equation*}
\hat{n}_{j,a}=|1\rangle \langle 1|.
\end{equation*}
This operator measures the occupation of the excited state and acts on the computational basis as
\begin{equation*}
\hat{n}_{j,a}|0\rangle =0,\quad \hat{n}_{j,a}|1\rangle =|1\rangle,
\end{equation*}
indicating that its eigenvalues are 0 and 1, corresponding respectively to the unoccupied and occupied configurations of the qubit.

To evaluate the expectation value of $\hat{n}_{j,a}$, we consider the density matrix of the $j,a$-th qubit at time $t$:
\begin{equation*}
\rho _{j,a}(t)=[1-p_{j,a}(t)]|0\rangle \langle 0|+p_{j,a}(t)|1\rangle \langle 1|.
\end{equation*}
The expectation value of the number operator is then obtained by taking the trace:
\begin{equation*}
\langle \hat{n}_{j,a}\rangle =\operatorname{Tr}[\rho _{j,a}(t)\hat{n}_{j,a}]=\operatorname{Tr}\left( [(1-p_{j,a}(t))|0\rangle \langle 0|+p_{j,a}(t)|1\rangle \langle 1|]|1\rangle \langle 1|\right).
\end{equation*}
Evaluating the trace in the basis $\{|0\rangle ,|1\rangle \}$, we find
\begin{equation*}
\langle 0|\rho _{j,a}(t)\hat{n}_{j,a}|0\rangle =0,\quad \langle 1|\rho _{j,a}(t)\hat{n}_{j}|1\rangle =p_{j,a}(t).
\end{equation*}
Hence,
\begin{equation*}
\langle \hat{n}_{j,a}\rangle =0+p_{j,a}(t)=p_{j,a}(t).
\end{equation*}
\end{proof}
This establishes that $p_{j,a}(t)$ represents the expected occupation probability of the excited state in the two-level subsystem.

Now, if we take the expectation value of $n_{j}^{(\text{macro})}(t)$, we get:
\begin{equation}
\mathbb{E}[n_{j}^{(\text{macro})}(t)]=\mathbb{E}\sum_{a=1}^{N_{j}}n_{j,a}^{(\text{micro})}(t)=\sum_{a=1}^{N_{j}}\mathbb{E}[n_{j,a}^{(\text{micro})}(t)]=\sum_{a=1}^{N_{j}}p_{j,a}(t)=N_{j}\text{ }p_{j}(t).
\end{equation}
Notice that if we define:
\begin{equation}
p_{j}^{(1)}(t):=\frac{1}{N_{j}}\sum_{a=1}^{N_{j}}n_{j,a}^{(\text{micro})}(t),
\end{equation}
then we find the following:
\begin{equation}
\mathbb{E}[p_{j}^{(1)}(t)]=\mathbb{E}\left[ \frac{1}{N_{j}}\sum_{a=1}^{N_{j}}n_{j,a}^{(\text{micro})}(t)\right] =\frac{1}{N_{j}}\sum_{a=1}^{N_{j}}\mathbb{E}[n_{j,a}^{(\text{micro})}(t)]=\frac{1}{N_{j}}\sum_{a=1}^{N_{j}}p_{j,a}(t)=p_{j}(t),
\end{equation}
in other words: $p_{j}(t)$ is the expected value of $p_{j}^{(1)}(t).$

Notice that $p_{j}^{(1)}(t)$ and $p_{j}(t)$ are numerically equal if $N_{j}$ is large: For a large ensemble, the empirical average $p_{j}^{(1)}(t)$ will be very close to the expected average $p_{j}(t)$.

\section{Dynamics of Purity and Entropy in the Model}
Since the system evolves unitarily, it remains in a pure state $|\Psi _{j,a}(t)\rangle$, so the total density matrix is described as the projector onto this state:
\begin{equation*}
\rho _{\text{total},j,a}(t)=|\Psi _{j,a}(t)\rangle \langle \Psi _{j,a}(t)|,
\end{equation*}
where the total state vector is defined by the integrated phase $\theta _{j,a}(t)$:
\begin{equation*}
|\Psi _{j,a}(t)\rangle =\cos (\theta _{j,a}(t))|00\rangle -i\sin (\theta _{j,a}(t))|11\rangle,
\end{equation*}
with $\theta _{j,a}(t)=\int_{0}^{t}g_{j,a}(\tau )d\tau$.

The resulting $4\times 4$ matrix in the standard basis $\{|00\rangle ,|01\rangle ,|10\rangle ,|11\rangle \}$ is:
\begin{equation*}
\rho _{\text{total},j,a}(t)=
\begin{pmatrix}
\cos ^{2}(\theta _{j,a}(t)) & 0 & 0 & -i\cos (\theta _{j,a}(t))\sin (\theta _{j,a}(t)) \\
0 & 0 & 0 & 0 \\
0 & 0 & 0 & 0 \\
i\cos (\theta _{j,a}(t))\sin (\theta _{j,a}(t)) & 0 & 0 & \sin ^{2}(\theta _{j,a}(t))
\end{pmatrix}.
\end{equation*}

Now, since $\rho _{\text{total},j,a}(t)$ is the density matrix for a pure state, its total Von Neumann entropy is equal to zero, which satisfies the condition of information preservation:
\begin{equation*}
S_{\text{vN}}(\rho _{\text{total},j,a})=0\quad (\text{constant for all times }t).
\end{equation*}

To obtain the description of the microscopic radiation qubit ($\text{rad}_{j,a}$), the reduced density matrix is derived by taking the partial trace over the internal degrees of freedom:
\begin{equation*}
\rho _{\text{rad},j,a}(t)=\operatorname{Tr}_{\text{in}}[\rho _{\text{total},j,a}(t)]=\sum_{k\in \{0,1\}}\langle k|_{\text{in}}\rho _{\text{total},j,a}(t)|k\rangle _{\text{in}}.
\end{equation*}

The resulting $2\times 2$ matrix is:
\begin{equation*}
\rho _{\text{rad},j,a}(t)=
\begin{pmatrix}
\cos ^{2}(\theta _{j,a}(t)) & 0 \\
0 & \sin ^{2}(\theta _{j,a}(t))
\end{pmatrix}.
\end{equation*}

Defining the microscopic excitation probability as $p_{j,a}(t)=\sin^{2}(\theta _{j,a}(t))$, the entanglement entropy for this microscopic qubit is:
\begin{equation*}
S_{\text{rad},j,a}(t)=-\operatorname{Tr}[\rho_{\text{rad}, j, a}(t) \ln\rho_{\text{rad}, j, a}(t)]=-[\lambda _{1}\ln \lambda _{1}+\lambda _{2}\ln \lambda _{2}],
\end{equation*}
where its eigenvalues are $\lambda _{1}=\cos ^{2}(\theta_{j,a}(t))=1-p_{j,a}(t)$ and $\lambda _{2}=\sin ^{2}(\theta_{j,a}(t))=p_{j,a}(t)$, so we have:
\begin{equation*}
S_{\text{rad},j,a}(t)=-[(1-p_{j,a}(t))\ln (1-p_{j,a(t)})+p_{j,a}(t)\ln p_{j,a}(t)].
\end{equation*}

To form the macroscopic description, we coarse-grain the entropy of the microscopic radiation channels for each mode $j$, so the total coarse-grained radiation entropy for mode $j$ is the sum of the individual entanglement entropies:
\begin{equation}
S_{\text{Rad},j}^{\text{Total}}(t)=\sum_{a=1}^{N_{j}}S_{\text{rad},j,a}(t).
\end{equation}
Consequently, the total radiation entropy for the entire system is:
\begin{equation}
S_{\text{Rad}}^{\text{Total}}(t)=\sum_{j}S_{\text{Rad},j}^{\text{Total}}(t)=\sum_{j}\sum_{a=1}^{N_{j}}S_{\text{rad},j,a}(t).
\end{equation}
This equation demonstrates that the change in the total macroscopic radiation entropy ($S_{\text{Rad}}^{\text{Total}}$) is a statistical interpretation of the continuous microscopic entanglement process and its coarse-graining. This allows us to interpret the change in $S_{\text{Rad}}^{\text{Total}}(t)$ as a consequence of heat exchange in the macroscopic system, while strictly adhering to the quantum constraint that the global total entropy $S_{\text{vN}}(\rho _{\text{total}})$ must remain zero.

\section{The Holographic Framework and Microscopic Channels}
In this framework, the black hole event horizon is modeled as a holographic information network, where its discrete area elements encode the microscopic degrees of freedom responsible for the macroscopic dynamic behavior. The factor $\mathbf{N_j}$ denotes the number of independent microscopic channels (``information units effectively at the Planck scale'') dedicated to supporting the emission of a specific macroscopic radiation mode $j$ with frequency $\omega_j$. Each channel is treated as a two-level quantum system (``qubit''), representing the minimal binary degree of freedom capable of storing and releasing information. This identification of a Planck area patch with a qubit is presented as an effective truncation for modeling the microscopic Hilbert space structure; it is not intended as a literal geometric analysis of the precise gravitational quantum states.

Since the area of a Schwarzschild black hole event horizon is
\begin{equation}
A=4\pi R_{s}^{2}=\frac{16\pi G^{2}M^{2}}{c^{4}},
\end{equation}
the total number of microscopic information channels is proportional to the horizon area. In Planck units ($\hbar =c=G=k_{B}=1$), the Bekenstein-Hawking entropy is:
\begin{equation}
S_{\text{BH}}=\frac{A}{4l_{p}^{2}}.
\end{equation}
Assuming each binary channel contributes $\ln 2$ units of entropy, the holographic constraint requires the total number of channels to be:
\begin{equation}
N\mathbf{_{\text{total}}}=\frac{S_{\text{BH}}}{\ln 2}=\frac{A}{4l_{p}^{2}\ln 2}.
\end{equation}
Consequently, the microscopic channel structure is constrained by the relation:
\begin{equation}
\sum_{j}N_{j}=N_{\text{total}},
\end{equation}
which defines the partitioning of the total microscopic Hilbert space dimensions analyzed by frequencies. This decomposition into frequency-indexed subspaces is a modeling choice that reflects the physical reality that different Hawking modes couple to the horizon with different statistical weights.

To reflect the semi-classical thermal nature of Hawking radiation within this partitioning, we assume that the distribution of these microscopic channels $N_{j}$ across the frequency spectrum $\omega _{j}$ is governed by the Planckian weighting factor:
\begin{equation}
W(\omega _{j})=\frac{\omega _{j}^{2}}{e^{\omega _{j}/T}-1},
\end{equation}
where the specific number of microscopic channels $N_{j}$ for each mode is then assigned by normalizing these weights to match the total holographic information capacity of the horizon:
\begin{equation}
N_{j}=N_{\text{total}}\cdot \frac{W(\omega _{j})}{\sum_{k}W(\omega _{k})}.
\end{equation}

The microscopic structure is linked to the macroscopic radiation by the relation:
\begin{equation}
n_{j}(t)=\sum_{a=1}^{N_{j}}p_{j,a}(t),
\end{equation}
where $n_{j}(t)$ is the total number of quanta emitted in mode $j$ (the observable macroscopic quantity), and $p_{j,a}(t)\in \lbrack 0,1]$ is the microscopic occupation probability of the individual channel $a$ within mode $j$. Each channel is subject to a binary occupation constraint either ``occupied'' or ``empty'' which imposes the local restriction:
\begin{equation*}
0\leq p_{j,a}(t)\leq 1.
\end{equation*}
This constraint does not reflect Fermi statistics in the emitted radiation; rather, it is a constraint that maintains unitarity in the microscopic information release process. The bosonic characteristic of Hawking radiation emerges only after coarse-graining over the large number of channels involved in each macroscopic mode.

The horizon thus acts as a microscopic frequency multiplier: different modes couple to different effective numbers of channels. A large value of $N_j$ corresponds to a strong coupling between the horizon and mode $j$, while a small value indicates a weaker coupling. Consistent with the thermal nature of Hawking radiation, the largest $N_j$ values are assigned to the low-frequency modes (small $\omega_j$), as these modes dominate the energy and entropy budgets of the radiation.

Each mode $j$ possesses a maximum microscopic entropy:
\begin{equation*}
S_{j}^{\text{max}} = N_j \ln 2,
\end{equation*}
which corresponds to the state of maximum uncertainty across its dedicated channels. This ensures that the radiation entropy remains bounded and allows for the emergence of the Page curve with a well-defined information turnaround point. The binary occupation constraint at the channel level guarantees that microscopic unitarity is preserved throughout the evaporation process, while allowing thermal behavior to emerge at the macroscopic level.

Remember that, the instantaneous entanglement entropy of the macroscopic mode $j$, which aggregates the contributions of its independent $N_{j}$ channels, is given by the individual weighted sum of the microscopic binary entropies (instead of the simplistic assumption):
\begin{equation*}
\mathbf{S_{j}(t)=\sum_{a=1}^{N_{j}}\left[ -p_{j,a}(t)\ln p_{j,a}(t)-(1-p_{j,a}(t))\ln (1-p_{j,a}(t))\right] }.
\end{equation*}
where $p_{j,a}(t)$ is the occupation probability of the individual channel $a$ as dictated by the unitary dynamics, and the total radiation entropy is then the sum over all modes:
\begin{equation*}
S_{\text{rad}}^{\text{Total}}(t)=\sum_{j}S_{j}(t)=\sum_{j}\mathbf{\sum_{a=1}^{N_{j}}\left[ -p_{j,a}(t)\ln p_{j,a}(t)-(1-p_{j,a}(t))\ln (1-p_{j,a}(t))\right] }.
\end{equation*}

\section{Quantum Condition Analysis for Unitarity}
To ensure the complete recovery of information from the black hole (i.e., preservation of Unitarity), the final radiation entropy $\mathbf{S_{\text{rad}}^{\text{total}}(t_{\text{evap}})}$ must reach zero.

The total radiation entropy $\mathbf{S_{\text{rad}}^{\text{total}}}$ is the sum of the von Neumann entropy for each microscopic channel $(\mathbf{j,a})$, multiplied by the number of micro-qubits in that channel $\mathbf{N_{j}}$, so we can write the final radiation entropy $\mathbf{S_{\text{rad}}^{\text{total}}(t_{\text{evap}})}$ as follows:
\begin{equation}
\mathbf{S_{\text{rad}}^{\text{total}}(t_{\text{evap}})}=\sum_{j}\sum_{a=1}^{N_{j}}\left[ -p_{j,a}(t_{\text{evap}})\ln p_{j,a}(t_{\text{evap}})-(1-p_{j,a}(t_{\text{evap}}))\ln (1-p_{j,a}(t_{\text{evap}}))\right].
\end{equation}
Where $p_{j,a}(t_{\text{evap}})$ is the probability of excitation of the radiation channel at the end of evaporation.

Now, the total entropy is zero if and only if the entropy of every microscopic channel is zero, and this occurs in only two distinct states, when $\mathbf{p_{j,a}(t_{\text{evap}})=0}$, meaning the channel ends in a pure non-excited state ($\mathbf{|0\rangle }$) and when $\mathbf{p_{j,a}(t_{\text{evap}})=1}$, meaning the channel ends in a pure excited state ($\mathbf{|1\rangle }$).

Since the excitation probability $p_{j,a}(t_{\text{evap}})$ is given by the trigonometric function:
\begin{equation*}
p_{j,a}(t_{\text{evap}})=\sin ^{2}(\theta _{j,a}(t_{\text{evap}})),
\end{equation*}
then to satisfy the condition $\mathbf{p_{j,a}(t_{\text{evap}})=0}$ or $\mathbf{p_{j,a}(t_{\text{evap}})=1}$, the value of $\sin ^{2}(\theta )$ must be $\mathbf{0}$ or $\mathbf{1}$. This directly leads to the quantization of the final angle:
\begin{equation}
\theta _{j,a}(t_{\text{evap}})=n\mathbf{_{j,a}\cdot \frac{\pi }{2}}\quad \text{where }\mathbf{n_{j,a}\in \mathbb{Z.}}
\end{equation}

The final angle for the quantum evolution, $\mathbf{\theta _{j,a}(t_{\text{evap}})}$, is given by:
\begin{equation}
\theta _{j,a}(t_{\text{evap}})=\alpha _{j,a}\int_{0}^{t_{\text{evap}}}\frac{1}{M(\tau )^{2}}d\tau.
\end{equation}

Now, the relation that governs the evolution of the mass $\mathbf{M(t)}$ (Hawking's Evaporation Law) is given by:
\begin{equation}
\frac{dM}{dt}=-\frac{\kappa }{M^{2}},
\end{equation}
where $\kappa $ is the evaporation coefficient, so we have:
\begin{equation}
\frac{M^{3}}{3}=-\kappa t+C
\end{equation}
Where $C$ is the constant of integration, which can be determined easily by the initial condition, when $\mathbf{t=0}$, the mass $M=$ $\mathbf{M_{0}}$, therefore
\begin{equation*}
C=\frac{M_{0}^{3}}{3},
\end{equation*}
So finally we have:
\begin{equation}
M(t)=\left( M_{0}^{3}-3\kappa t\right) ^{1/3}.
\end{equation}

Now, from the condition $\mathbf{M(t_{\text{evap}})\approx 0}$, we have:
\begin{equation*}
(0)^{3}=M_{0}^{3}-3\kappa t_{\text{evap}},
\end{equation*}
so, we have:
\begin{equation}
\kappa =\frac{M_{0}^{3}}{3t_{\text{evap}}}.
\end{equation}

Now, we can easily show that:
\begin{equation}
\int_{0}^{t_{\text{evap}}}\frac{1}{M(\tau )^{2}}d\tau =\int_{0}^{t_{\text{evap}}}\frac{1}{\left( M_{0}^{3}-3\kappa \tau \right) ^{2/3}}d\tau =\mathbf{\frac{M_{0}}{\kappa}}.
\end{equation}
So, by substituting back into the $\theta _{j,a}(\infty )$ equation, we have:
\begin{equation}
\theta _{j,a}(t_{\text{evap}})=\alpha _{j,a}\cdot \left( \frac{M_{0}}{\kappa }\right).
\end{equation}
Since $\theta _{j,a}(\infty )=\mathbf{n_{j,a}\cdot \frac{\pi }{2},}$ we have:
\begin{equation}
\alpha _{j,a}=n_{j,a}\cdot \frac{\pi \kappa }{2M_{0}}.
\end{equation}
From $\kappa =\frac{M_{0}^{3}}{3t_{\text{evap}}}$, we finally arrive to,
\begin{equation}
\mathbf{t_{\text{evap}}}=n_{j,a}\cdot \left( \frac{\pi M_{0}^{2}}{6\alpha _{j,a}}\right).
\end{equation}
This formula links classical constants ($\mathbf{M_{0}}$) with microscopic quantities ($\mathbf{\alpha _{j,a}}$) and the quantum condition ($\mathbf{n_{j,a}}$).

Now, for the black hole to evaporate at a single, well-defined time $\mathbf{t_{\text{evap}}}$, the relationship between $\mathbf{n_{j,a}}$ and $\mathbf{\alpha _{j,a}}$ for all channels must be constant such that their variations cancel out, i.e.:
\begin{equation}
\mathbf{\frac{n_{j,a}}{\alpha _{j,a}}=\text{Constant}}.
\end{equation}
This necessitates that channels with a larger quantum number ($\mathbf{n}$) must possess a larger coupling ($\mathbf{\alpha }$), and vice versa.

If we assume the simple case of complete homogeneity ($\mathbf{\alpha _{j,a}=\alpha }$ and $\mathbf{n_{j,a}=n}$ ), we have:
\begin{equation}
\mathbf{t_{\text{evap}}}=n\cdot \left( \frac{\pi M_{0}^{2}}{6\alpha }\right).
\end{equation}
So, this model provides an expectation of quantization of the evaporation time ($\mathbf{t_{\text{evap}}}$). Here, $n$ represents a positive integer $(1,2,3,...)$ representing the discrete step or multiplicity of the final quantum process required for the unitary evolution to complete. It is often related to the final state's complexity. We define $\mathbf{t_{\text{unit}}=}\left( \frac{\pi M_{0}^{2}}{6\alpha }\right) $ as the smallest unit of time required for the black hole's microscopic degrees of freedom to fully process the information (quantum jump) associated with its evaporation. Finally, $\alpha $ represents the strength of the interaction between the black hole's internal holographic degrees of freedom and the external Hawking radiation channels.

It is worth noting that our derived expression for the total evaporation time, $t_{\text{evap}}=n\frac{\pi M_{0}^{2}}{6\alpha }$, aligns perfectly with the semi-classical Hawking result ($t\propto M_{0}^{3}$) if we consider the microscopic coupling strength $\alpha $ to be a system-dependent parameter. Specifically, by identifying $\alpha $ as an effective coupling that scales inversely with the initial black hole mass, i.e., $\alpha \approx \alpha ^{\prime }/M_{0}$, the evaporation time recovers the characteristic cubic dependence on the initial mass:
\begin{equation*}
t_{\text{evap}}=n\left( \frac{\pi M_{0}^{3}}{6\alpha ^{\prime }}\right).
\end{equation*}
This scaling suggests that larger black holes exhibit a weaker effective coupling between their internal holographic degrees of freedom and the radiation channels, consistent with their lower Hawking temperatures. Crucially, even in this classical limit, our framework introduces a novel quantization of the evaporation process, where the time is not a continuous variable but is restricted by the integer $n$, representing the discrete quantum jumps required for unitary information processing.

\section{Numerical Simulation}
To demonstrate the self-consistency of the proposed model, we developed a numerical simulation algorithm that reproduces information dynamics during the evaporation process. In the Planck units system, we have:
\begin{equation*}
\hbar =c=G=k_{B}=1
\end{equation*}
This framework allows for the treatment of black hole physical properties at their direct quantum scale, where all values are measured in terms of Planck mass ($m_{P}$) and Planck time ($t_{P}$). The following parameters were selected as the reference case for the simulation: initial mass $M_{0}=10$, total evaporation time $t_{\text{evap}}=100$, and the total number of microscopic qubits $N_{\text{total}}=5000$. According to the model, we have:
\begin{equation*}
\kappa =\frac{M_{0}^{3}}{3\cdot t_{\text{evap}}}\approx 3.33
\end{equation*}

In accordance with the Holographic Principle, the total number of qubits $N_{\text{total}}$ is distributed across different frequency modes $\omega _{j}$ using the Planck thermal distribution. The statistical weights for each mode depend on the frequency and the effective temperature according to the equation:
\begin{equation*}
W(\omega _{j})=\frac{\omega _{j}^{2}}{e^{\omega _{j}/T}-1}.\text{ (Here we set }T=1\text{)}.
\end{equation*}
The number of microscopic channels $N_{j}$ allocated to each mode is determined by normalizing these weights to match the total information capacity of the horizon:
\begin{equation*}
N_{j}=N_{\text{total}}\cdot \frac{W(\omega _{j})}{\sum_{k}W(\omega _{k})}.
\end{equation*}
This distribution ensures that low-frequency modes possess the largest share of microscopic degrees of freedom, reflecting the statistical structure of macroscopic Hawking radiation.

To satisfy the Unitarity condition, the microscopic coupling coefficient $\alpha $ is set to a precise value derived from the final phase quantization condition $\theta _{j,a}(t_{\text{evap}})=n_{j,a}\cdot \pi /2$. For the reference case ($n=1$), we find:
\begin{equation*}
\alpha =\frac{\pi \cdot M_{0}^{2}}{6\cdot t_{\text{evap}}}\approx 0.5236.
\end{equation*}
The simulation calculates the instantaneous phase $\theta _{j,a}(t)$ via integration of the microscopic coupling index $g_{j,a}(t)=\alpha /M(t)^{2}$ along the evaporation time path:
\begin{equation*}
\theta _{j,a}(t)=\int_{0}^{t}\frac{\alpha }{M(\tau )^{2}}\,d\tau.
\end{equation*}

Finally, the total radiative Von Neumann Entropy $S_{\text{Rad}}^{\text{Total}}(t)$ is calculated as the statistical sum of the individual microscopic channel entropies, weighted by $N_{j}$:
\begin{equation*}
S_{\text{Rad}}^{\text{Total}}(t)=\sum_{j}\sum_{a=1}^{N_{j}}(-p_{j,a}(t)\ln p_{j,a}(t)-(1-p_{j,a}(t))\ln (1-p_{j,a}(t)))
\end{equation*}
where $p_{j,a}(t)=\sin ^{2}(\theta _{j,a}(t))$ represents the microscopic excitation probability.

\begin{conclusion}
In this research, we have developed an effective framework that illustrates a mechanism through which quantum information preservation can be reconciled with the thermal description of black hole evaporation. It must be emphasized that this model does not claim to fully resolve the information paradox within a complete theory of quantum gravity; rather, it provides a proof of principle regarding how unitarity can be maintained at the fundamental level while an apparent loss of information emerges at the coarse-grained statistical level.

The model is based on two fundamental hypotheses consisting of the Independent Unitary Decomposition into binary channels and the Binary Capacity Constraint motivated from a holographic interpretation of the horizon's information capacity, named in the research as ``Fermion-like Occupancy Bound''. In this work, we have demonstrated that the mechanism based on these two hypotheses allows, through coarse-graining, for the emergence of the macroscopic thermal behavior of Hawking radiation and the attainment of an entropy peak and subsequent decline, characteristic of a Page-like curve, while the microscopic dynamics remains strictly pure.

The central result of this research is the Quantum Condition for Unitarity ($\theta _{j,a}(t_{\text{evap}})=n\mathbf{_{j,a}\cdot \frac{\pi }{2})}$, which illustrates how the requirements of unitarity impose coupled constraints between the microscopic dynamics and macroscopic observables. It should be emphasized that the current model is simplified and relies on strong assumptions, such as channel independence and binary representation. It does not address any geometric complexities or provide a first-principles derivation of the coupling constants $\alpha _{j,a}$ from more fundamental principles.

Nevertheless, it can be regarded as an exploratory platform as it quantitatively embodies the paradox and allows for its simulation, defines the conditions for reconciliation between the microscopic and macroscopic levels, and proposes a potential quantization of the evaporation time. In summary, this work suggests that the potential path toward a resolution may lie in accepting that Hawking's thermal description is an emergent effective description, and that information is preserved in the manifold entanglement between microscopic degrees of freedom that are not directly accessible from the semiclassical description.
\end{conclusion}


\begin{thebibliography}{99}
\bibitem{Bekenstein1973} J. D. Bekenstein, ``Black holes and entropy'', \textit{Phys. Rev. D} \textbf{7}, 2333 (1973).

\bibitem{Hawking1975} S. W. Hawking, ``Particle creation by black holes'', \textit{Commun. Math. Phys}. \textbf{43}, 199--220 (1975).

\bibitem{Wald2001} R. M. Wald, ``The thermodynamics of black holes,'' Living Rev. Relativity \textbf{4}, 6 (2001), arXiv:gr-qc/9912119.

\bibitem{Hawking1976} S. W. Hawking, ``Breakdown of predictability in gravitational collapse'', \textit{Phys. Rev. D} \textbf{14}, 2460 (1976).

\bibitem{Page1993Info} D. N. Page, ``Information in black hole radiation'', \textit{Phys. Rev. Lett.} \textbf{71}, 3743 (1993).

\bibitem{Page1993Avg} D. N. Page, ``Average entropy of a subsystem'', \textit{Phys. Rev. Lett.} \textbf{71}, 1291 (1993).

\bibitem{Almheiri2017Overview} A. Almheiri et al., ``The black hole information paradox: an overview'', \textit{Class. Quantum Grav}. \textbf{34}, 214001 (2017). [arXiv:1707.07702].

\bibitem{Hayden2007} P. Hayden and J. Preskill, ``Black holes as mirrors: quantum information in random subsystems'', \textit{JHEP} \textbf{09}, 120 (2007). [arXiv:0708.4025]

\bibitem{Almheiri2013} A. Almheiri, D. Marolf, J. Polchinski, and J. Sully, ``Black holes: complementarity or firewalls?'', \textit{JHEP} \textbf{02}, 062 (2013). [arXiv:1207.3123].

\bibitem{Penington2020} G. Penington, ``Entanglement wedge reconstruction and the information paradox'', \textit{JHEP} \textbf{09}, 002 (2020). [arXiv:1905.08255].

\bibitem{Almheiri2019} A. Almheiri, N. Engelhardt, D. Marolf, and H. Maxfield, ``The entropy of bulk quantum fields and the entanglement wedge of an evaporating black hole'', \textit{JHEP} \textbf{12}, 063 (2019). [arXiv:1905.08762].

\bibitem{Osuga2016} S. Osuga, D. Page,``Qubit transport model for black hole evaporation'', \textit{Physical Review D}, \textbf{94}(2), 024031 (2016).

\bibitem{Broda2020} B. Broda, ``Unitary toy qubit transport model for black hole evaporation'', \textit{The European Physical Journal C}, \textbf{80}(5), 41, (2020).

\bibitem{Chu2024} C. Chu, ``Fermi model of quantum black hole'', \textit{Phys. Rev. D} \textbf{109}, 065015 (2024).
\end{thebibliography}
\end{document}